\documentclass[conference]{IEEEtran}
\IEEEoverridecommandlockouts
\usepackage{cite}
\usepackage{amsmath,amssymb,amsfonts}
\usepackage{graphicx}
\usepackage{textcomp}
\usepackage{xcolor}
\usepackage{amsthm}
\usepackage[ruled,vlined,linesnumbered]{algorithm2e}
\usepackage{float}
\usepackage[caption=false, font=footnotesize]{subfig}
\usepackage{multirow}
\usepackage[hidelinks]{hyperref}

\def\BibTeX{{\rm B\kern-.05em{\sc i\kern-.025em b}\kern-.08em
    T\kern-.1667em\lower.7ex\hbox{E}\kern-.125emX}}
\begin{document}

% \title{Interval-Filtering ANNS for Containment and Overlap: A Unified Design Framework\\
% \title{Unified Dominance Graph for Two-Bound Interval-Filtered Approximate Nearest Neighbor Search
\title{Unified Dominance Graph for Interval-Predicate Approximate Nearest Neighbor Search
% {\footnotesize \textsuperscript{*}Note: Sub-titles are not captured in Xplore andshould not be used}
}

\author{
\IEEEauthorblockN{
Kwun Hang Lau\textsuperscript{1},
Ruiyuan Zhang\textsuperscript{2},
Elton Chun-Chai Li\textsuperscript{1},\\
Wun Yu Chan\textsuperscript{1},
Xiaojun Cheng\textsuperscript{3},
Xiaofang Zhou\textsuperscript{1}}
\IEEEauthorblockA{
\textsuperscript{1}\,The Hong Kong University of Science and Technology;
\textsuperscript{2}\,Hong Kong Generative AI R\&D Center;\\
\textsuperscript{3}\,China Unicom (Hong Kong) Operation Ltd\\
\{khlaube, cceli, wychanbu\}@connect.ust.hk;
zry@hkgai.org;
chengxj31@chinaunicom.cn;
zxf@cse.ust.hk}
}
% \author{\IEEEauthorblockN{Kwun Hang Lau}
% \IEEEauthorblockA{\textit{Department of Computer Science and Engineering} \\
% \textit{The Hong Kong University of Science and Technology}\\
% Hong Kong \\
% khlaube@connect.ust.hk}
% \and
% \and
% \and
% \and
% \and

\maketitle
\newtheorem{theorem}{Theorem}
\newtheorem{lemma}{Lemma}
\newtheorem{example}{Example}

\begin{abstract}
Approximate Nearest Neighbor Search (ANNS) is a core primitive for unstructured data retrieval. Real-world applications—such as temporal databases, financial data analysis, and retrieval-augmented generation—often require hybrid queries whose valid objects are constrained by continuous interval attributes, such as lifespans or price ranges. We study Interval-Predicate ANNS (IPANNS), where validity is determined by a predicate between an object interval and a query interval. Existing range-filtering ANNS (RFANNS) methods are designed for single-dimensional scalar filters, but interval predicates such as containment and overlap rely on two coupled endpoint constraints. Treating endpoints as independent scalar attributes can incur large intersection overhead, while containment-specific methods lack a generalized indexing abstraction.

In this paper, we propose the Unified Dominance Graph (UDG), a graph-indexing framework for the closed two-bound conjunctive fragment of IPANNS. For a chosen interval predicate, UDG maps object and query endpoints into a normalized two-dimensional dominance space and builds a dominance-labeled graph over the transformed coordinates. Containment, overlap, and other supported endpoint-bound predicates therefore reuse the same construction and search algorithms after semantic mapping, while each UDG instance remains tied to its selected predicate. UDG compresses query-state-specific proximity graphs into one compact index.
To improve graph search under restrictive interval filters, we add validity-preserving patch edges that provide routing choices when few objects remain valid.
Extensive evaluations on standard benchmarks and real-world datasets show that UDG achieves stable query performance across multiple interval relations and workloads, significantly outperforming existing hybrid search baselines while maintaining low indexing overhead.
\end{abstract}

\begin{IEEEkeywords}
high-dimensional data, approximate nearest neighbor search
\end{IEEEkeywords}

\section{Introduction}
The rapid adoption of approximate nearest neighbor search (ANNS) has made vector similarity search a core primitive for unstructured data retrieval. However, real-world database queries rarely rely on vector similarity alone. Users often issue hybrid queries that combine embedding similarity with structured metadata constraints. This requirement has motivated extensive research on range-filtering ANNS (RFANNS), which retrieves the top-$k$ nearest vectors whose associated scalar attributes fall within a specified numeric range~\cite{SeRF-2024,dynamic_serf,window-filter-2024-RFANNS,DIGRA-2025-rfanns,RangePQ-2025-rfanns,iRangeGraph-2024-rfanns}. Although state-of-the-art RFANNS indexes achieve strong performance for one-dimensional scalar attributes, many real-world objects cannot be adequately represented by a single point-valued attribute.

Interval attributes arise naturally in temporal databases, multimedia retrieval, financial data analysis, and temporal retrieval-augmented generation. For example, an event or video segment is described by a duration $[s_i,t_i]$; a stock may fluctuate within a daily price range; and a document, fact, or subscription may be valid only during a specific time interval. Recent work Hi-PNG~\cite{hi-png-2025} has recognized the importance of interval-aware hybrid search, but it primarily focuses on containment, where a data interval must be fully contained in the query interval.

Containment, however, is only one type of interval relation. In many applications, users need overlap queries, where an object is valid if its interval intersects the query interval. For example, a time-sensitive retrieval system may need to find events, documents, or active subscriptions that overlap a specified month. More generally, a broad class of interval predicates can be expressed as two closed endpoint inequalities. Containment requires $s_i \geq s_q$ and $t_i \leq t_q$, whereas overlap requires $t_i \geq s_q$ and $s_i \leq t_q$. These predicates differ in which endpoints are compared and in the directions of their inequalities, leading to different dominance geometries over the interval endpoint space.

Despite this practical need, efficiently supporting multiple interval predicates is still challenging. Existing RFANNS methods exploit a single total order over scalar attributes. Interval predicates, in contrast, induce a two-dimensional endpoint space: sorting by the start endpoint ignores the end constraint, while sorting by the end endpoint ignores the start constraint. An alternative is to treat an interval $[s_i,t_i]$ as two independent scalar attributes and reduce interval filtering to two RFANNS problems. However, this approach separates two coupled endpoint constraints and can introduce substantial candidate-generation and intersection overhead. Similarly, containment-oriented interval indexes do not directly provide a general indexing abstraction for overlap and other closed two-bound relations.

To address this gap, we study Interval-Predicate ANNS (IPANNS), where each query combines vector similarity with a predicate between an object interval and a query interval. We propose Unified Dominance Graph (UDG), a graph-indexing framework for the closed two-bound conjunctive fragment of IPANNS. The key idea is to compile a chosen interval predicate into a normalized two-dimensional dominance predicate over transformed endpoint coordinates. Once this semantic mapping is fixed, UDG uses the same edge-labeling, construction, and search procedures regardless of whether the original predicate is containment, overlap, or another supported endpoint-bound predicate. This unifies the indexing principle across interval predicates while preserving interval-valid graph traversal.

In summary, we make the following contributions.
\begin{itemize}
    \item We formulate Interval-Predicate ANNS for closed two-bound conjunctive interval predicates, covering containment, overlap, and related endpoint-bound predicates under a unified problem setting.

    \item We propose Unified Dominance Graph (UDG), a unified graph index that maps different interval predicates into a common transformed coordinate space and supports them with the same framework.

    \item We prove that UDG provides a structural lossless compression of many query-specific proximity graphs into one labeled graph, and analyze its time and space complexity. We further design validity-preserving patch edges to improve navigability under restrictive interval filters.

    \item Extensive experiments on benchmark and real-world datasets show that UDG achieves superior performance for containment, overlap, and additional interval relations, while maintaining low index cost.
\end{itemize}

The source code is available at \href{https://github.com/kwunhang/UDG.git}{kwunhang/UDG.git}.
\section{Preliminary}
\subsection{Problem Definition}
\label{sec:problem-definition}

We first define Interval-Predicate Approximate Nearest Neighbor Search (IPANNS).
Let $\mathbb{R}^{d}$ be a $d$-dimensional vector space equipped with a distance metric $\delta$.
Given two vectors $u,v \in \mathbb{R}^{d}$, $\delta(u,v)$ measures their dissimilarity.
When the distance function is clear from the context, we omit it from the notation.

\textsc{Definition 1} (Nearest Neighbor Search)
\textit{Given a query vector $q \in \mathbb{R}^{d}$, an integer $k>0$, and a set $\mathcal{D}$ of vectors in $\mathbb{R}^{d}$, the $k$ nearest neighbors of $q$ in $\mathcal{D}$, denoted by $\mathrm{kNN}(q,\mathcal{D})$, is a subset $R \subseteq \mathcal{D}$ of $k$ vectors such that for any $u \in R$ and any $v \in \mathcal{D} \setminus R$, $\delta(u,q) \leq \delta(v,q)$.}

In this paper, each data vector is associated with a closed interval attribute, and queries impose interval predicate constraints during nearest neighbor search.
An interval-attributed vector is represented as $(v_i,[s_i,t_i])$, where $v_i \in \mathbb{R}^{d}$ is the embedding vector and $[s_i,t_i]$ is its interval attribute with $s_i \leq t_i$.
A query is represented as $Q=(q,[s_q,t_q],k,\rho)$, where $q \in \mathbb{R}^{d}$ is the query vector, $[s_q,t_q]$ is the query interval, $k$ is the result size, and $\rho$ specifies the interval predicate.

IPANNS allows the validity of an object to depend on an interval predicate between the object interval and the query interval.
This paper focuses on the predicate class supported by UDG: closed two-bound conjunctive relations.
Each relation is the conjunction of two endpoint comparisons.
Each comparison relates one data endpoint, either $s_i$ or $t_i$, to one query endpoint, either $s_q$ or $t_q$, using $\geq$ or $\leq$.
This scope excludes disjunctions, equality-only predicates, strict-only predicates, interval length or gap constraints, and relations requiring three or more endpoint inequalities.

The two primary workloads evaluated in this paper are containment and overlap.
For containment, a data interval must be fully contained in the query interval: $\rho_{\mathsf{con}}([s_i,t_i],[s_q,t_q]) \equiv (s_i \geq s_q) \wedge (t_i \leq t_q)$.
For overlap, a data interval must intersect the query interval: $\rho_{\mathsf{ovl}}([s_i,t_i],[s_q,t_q]) \equiv (t_i \geq s_q) \wedge (s_i \leq t_q)$.

\begin{example}[Interval predicate queries]
Consider four data intervals $A=[1,5]$, $B=[3,7]$, $C=[6,9]$, and $D=[8,12]$.
For a containment query $[2,10]$, the valid objects are $B$ and $C$.
For an overlap query $[4,7]$, the valid objects are $A$, $B$, and $C$.
\end{example}

Given a predicate $\rho$, we denote by $\mathcal{D}_{\rho}[s_q,t_q]\triangleq \{v_i \mid (v_i,[s_i,t_i])\in\mathcal{D},\rho([s_i,t_i],[s_q,t_q])\}$ the subset of vectors whose intervals satisfy the query predicate.

\textsc{Definition 2} (Interval-Predicate Nearest Neighbor Search)
\textit{Given a static dataset $\mathcal{D}$ of interval-attributed vectors, a query $Q=(q,[s_q,t_q],k,\rho)$ returns $\mathrm{kNN}(q,\mathcal{D}_{\rho}[s_q,t_q])$, i.e., the $k$ vectors in $\mathcal{D}_{\rho}[s_q,t_q]$ with the smallest distances to $q$.}

For simplicity of presentation and evaluation, we assume that each generated query has at least $k$ valid vectors in $\mathcal{D}_{\rho}[s_q,t_q]$.
The selectivity of an interval-predicate query is defined as $\sigma(Q) \triangleq |\mathcal{D}_{\rho}[s_q,t_q]|/|\mathcal{D}|$.
Smaller $\sigma(Q)$ means that fewer objects satisfy the interval predicate, so the query is more restrictive.
Restrictive interval predicates induce small valid subgraphs, where graph-based ANNS methods may suffer from poor navigability.

Due to the curse of dimensionality~\cite{curse_dim}, following a large body of existing research on approximate nearest neighbor search, this paper studies the approximate version of this task.

\textsc{Definition 3} (Interval-Predicate ANNS, IPANNS)
\textit{Given a static set $\mathcal{D}$ of interval-attributed vectors and a query $Q=(q,[s_q,t_q],k,\rho)$, where $\rho$ is a supported interval predicate, IPANNS aims to return a set $\mathrm{kANN}(q,\mathcal{D}_{\rho}[s_q,t_q]) \subseteq \mathcal{D}_{\rho}[s_q,t_q]$ of $k$ vectors, with recall $
    \frac{
    |\mathrm{kANN}(q,\mathcal{D}_{\rho}[s_q,t_q]) \cap \mathrm{kNN}(q,\mathcal{D}_{\rho}[s_q,t_q])|
    }{k}$.}

For ease of reference, the main notations and their meanings are summarized in Table~\ref{tab:symbols}, which will be referred to throughout the rest of the paper.

\begin{table}[t]
\centering
\footnotesize
\caption{Summary of symbols.}
\setlength{\tabcolsep}{2.5pt}
\begin{tabular}{@{}p{0.27\columnwidth}p{0.68\columnwidth}@{}}
\hline
Symbol & Meaning \\
\hline
$\mathcal{D}, n$ & Dataset and number of objects \\
$v_i$ & Embedding vector of object $i$ \\
$[s_i,t_i]$ & Closed interval of object $i$ \\
$Q$ & Query $(q,[s_q,t_q],k,\rho)$ \\
$q, k$ & Query vector and result size \\
$[s_q,t_q]$ & Query interval \\
$\delta(\cdot,\cdot)$ & Distance in embedding space \\
$\rho$ & Supported interval predicate \\
$\mathcal{D}_{\rho}[s_q,t_q]$ & Vectors satisfying $\rho$ \\
$\sigma(Q)$ & Query selectivity \\
$(X_i,Y_i)$ & Transformed data coordinates \\
$(x_q,y_q)$ & Transformed query coordinates \\
$U_X,U_Y$ & Distinct transformed coordinates \\
$(a,c), V(a,c)$ & Canonical dominance state and its valid object set \\
\hline
\end{tabular}
\label{tab:symbols}
\end{table}

The key difficulty is that interval-predicate filtering is not a scalar range filter.
Containment, overlap, and other closed two-bound conjunctive relations are determined by two coupled endpoint inequalities, so the valid set is not a contiguous slice over one ordered attribute.
The next subsection introduces edge-labeled compression as a general way to represent many query-specific proximity graphs, before we show how to instantiate it for interval predicates.

\subsection{Edge-Labeled Graph Compression for Filtered ANNS}
\label{sec:edge-labeled-compression}
A direct way to support query-efficient filtered ANNS is to build one ANNS index for every possible query filter condition. While this avoids the cost of evaluating invalid nodes during the search, the storage overhead is prohibitive. For a dataset of size $n$, even a simple filtering attribute can induce $O(n^2)$ distinct valid subsets, making individual index materialization costly.

\begin{algorithm}[t]
\caption{\textsc{Prune}$(o,ann,M)$}
\label{alg:prune}
\KwIn{$o$: a vector; $ann$: a set of $o$'s approximate nearest neighbors; $M$: the max number of neighbors to keep.}
\KwOut{$neighbors\subseteq ann$: $o$'s neighbors after pruning}

$neighbors \leftarrow \emptyset$\;
Sort $ann$ in ascending order of $\delta(o,u)$, with ties broken by object id\;
\ForEach{$u\in ann$ in sorted order}{
    $dominated \leftarrow \mathrm{false}$\;
    \ForEach{$w\in neighbors$}{
        \If{$\delta(o,w)<\delta(o,u)$ and $\delta(w,u)<\delta(o,u)$}{
            $dominated \leftarrow \mathrm{true}$\ and \textbf{break}\;
        }
    }
    \lIf{$dominated=\mathrm{false}$}{
        add $u$ to $neighbors$
    }
    \lIf{$|neighbors| \geq M$}{
        \textbf{break}
    }
}
\Return $neighbors$\;
\end{algorithm}

To resolve this, recent approaches utilize edge-labeled proximity graphs~\cite{SeRF-2024,dynamic_serf}. Instead of building multiple graphs, these methods compress them into a single unified structure.
For each node $u$, its neighbor list contains tuples of the form $(l,r,v,b,e)$, where $v$ is a neighboring vector. 
The tuple indicates that $v$ is available from $u$ for all range queries $[s_q,t_q]$, with $s_q \in (l,r]$ and $t_q \in [b,e)$. 
For example, a tuple $(2,9,v,11,15)$ is activated by query range $[5,13]$, but not by $[10,13]$ or $[5,16]$. 

The structural property desired of such a graph is lossless compression: for every filter condition, the active edge set should match the dedicated graph that would be built over the corresponding valid vectors.
This guarantee assumes exact construction-time neighbor search, while query-time ANNS remains approximate in practice. Prior work established such compression for scalar range filters~\cite{SeRF-2024}. Interval filtering requires two coupled inequalities, and extending this guarantee to 2D cross-ordered domains demands a new formal proof. In Section~\ref{sec:exact-construction-setting}, we define the exact construction setting used by UDG.

\section{Unified Dominance Abstraction}
Section~\ref{sec:problem-definition} defines IPANNS using semantic interval predicates.
This section introduces the abstraction that makes UDG relation-independent after a one-time semantic mapping step.
The key observation is simple: each supported interval relation consists of two endpoint inequalities, and the direction of each inequality can be normalized by selecting and, when necessary, negating the corresponding endpoint coordinate.
After this normalization, all supported relations are seen by the index as the same two-dimensional dominance predicate.

\subsection{From Scalar Range Filtering to Interval Dominance}
\label{sec:2d-dominance-challenge}
Existing RFANNS methods are designed around a single ordered scalar attribute.
Graph-based methods such as SeRF, Dynamic RFANNS, and iRangeGraph compress or assemble proximity graphs whose validity is defined over scalar query ranges~\cite{SeRF-2024,dynamic_serf,iRangeGraph-2024-rfanns}; dynamic tree-based methods such as DIGRA and PQ-based methods such as RangePQ still organize filtering around the same scalar range semantics~\cite{DIGRA-2025-rfanns,RangePQ-2025-rfanns}.
Their common premise is that the filter can be described by one total order: a query selects a contiguous scalar range, and index states or edge labels can be maintained along that order.
Interval filtering breaks this premise.
The interval predicates targeted by UDG are defined by two coupled endpoint inequalities, so choosing the start order ignores the end constraint, choosing the end order ignores the start constraint, and the valid set becomes a two-dimensional dominance region rather than a one-dimensional scalar segment.

This mismatch is why interval filtering cannot be obtained by a direct lift of scalar RFANNS.
Running two scalar filters and intersecting their outputs separates the endpoint constraints before graph search; the intersection may identify valid objects, but it does not provide the interval-valid proximity graph needed for navigable ANNS.
Replacing the scalar range structure with an interval tree or another exact reporting structure can enumerate valid intervals, but the resulting decomposition may span multiple nodes or branches and does not provide the single monotone lifetime over which existing graph labels and compression proofs are defined.
Therefore, UDG builds the graph around the native shape of interval predicates: after each supported relation is mapped to dominance coordinates, every query is a two-dimensional dominance state, and every graph edge is labeled by the set of dominance states in which that edge is valid.

\subsection{Mapping Relations to Dominance}
\label{sec:algebraic-mapping-rule}

The physical predicate used by UDG is
\begin{equation}
    X_i \geq x_q \quad \wedge \quad Y_i \leq y_q.
    \label{eq:target-dominance}
\end{equation}
Here $(X_i,Y_i)$ are transformed data coordinates, and $(x_q,y_q)$ are transformed query coordinates.
The transformation is applied before index construction, so all later graph operations are defined over transformed coordinates rather than over relation-specific interval operators.

Table~\ref{tab:expanded-predicate-mapping} lists the mappings used in this paper for representative relations.
Each row is a semantic mapping into the same dominance-indexing framework, not a separate graph algorithm.

\begin{table}[t]
\centering
\small
\footnotesize
\caption{Representative mappings from interval semantics to the normalized dominance predicate.}
\setlength{\tabcolsep}{2.4pt}
\begin{tabular}{lccccc}
\hline
Relation & Original relation & $X_i$ & $x_q$ & $Y_i$ & $y_q$ \\
\hline
Containment & $s_i \geq s_q \wedge t_i \leq t_q$ & $s_i$ & $s_q$ & $t_i$ & $t_q$ \\
Overlap & $t_i \geq s_q \wedge s_i \leq t_q$ & $t_i$ & $s_q$ & $s_i$ & $t_q$ \\
Query-within-data & $s_i \leq s_q \wedge t_i \geq t_q$ & $t_i$ & $t_q$ & $s_i$ & $s_q$ \\
Both boundaries after & $s_i \geq s_q \wedge t_i \geq t_q$ & $s_i$ & $s_q$ & $-t_i$ & $-t_q$ \\
Both boundaries before & $s_i \leq s_q \wedge t_i \leq t_q$ & $-s_i$ & $-s_q$ & $t_i$ & $t_q$ \\
\hline
\end{tabular}
\label{tab:expanded-predicate-mapping}
\end{table}

The mapping is not unique.
It only needs to assign one endpoint comparison to the $X$ axis and the other to the $Y$ axis so that the transformed predicate becomes Eq.~\eqref{eq:target-dominance}.
Signs are chosen only to make the two assigned comparisons match the directions in Eq.~\eqref{eq:target-dominance}.
Once a relation is mapped, UDG fixes the transformed coordinates and all later construction and search steps operate only in this dominance space.

The first two rows are the primary workloads evaluated in this paper.
Using the running example from the previous section, a containment query $[2,10]$ maps to $(x_q,y_q)=(2,10)$, so intervals $B=[3,7]$ and $C=[6,9]$ satisfy $X_i\geq x_q$ and $Y_i\leq y_q$.
For an overlap query $[4,7]$, the data coordinates become $(X_i,Y_i)=(t_i,s_i)$ and the transformed query is again a dominance point; intervals $A$, $B$, and $C$ satisfy Eq.~\eqref{eq:target-dominance}.
Axis assignment is part of the mapping.
For example, query-within-data places $t_i\geq t_q$ on the $X$ side and $s_i\leq s_q$ on the $Y$ side, so the index again sees only the normalized predicate.
The signed rows are also handled by the same rule.

\subsection{Canonical Query States}
\label{sec:canonical-query-space}

The transformed query point $(x_q,y_q)$ may lie between data endpoints.
For indexing, UDG only needs to distinguish query boundaries that change the valid set.
Let $U_X$ and $U_Y$ be the sorted distinct values among $\{X_i\}_{i=1}^{n}$ and $\{Y_i\}_{i=1}^{n}$, respectively.
Since only these transformed data coordinates can change the truth of $X_i\geq x_q$ or $Y_i\leq y_q$, UDG snaps the raw transformed query to canonical boundaries:
\begin{align}
    x_q^{+} &= \min\{x \in U_X \mid x \geq x_q\}, \nonumber\\
    y_q^{-} &= \max\{y \in U_Y \mid y \leq y_q\}.
    \label{eq:canonical-boundaries}
\end{align}
If either boundary is undefined, the valid set is empty.
Otherwise, canonicalization is exact rather than approximate.

\begin{lemma}[Canonical Query Equivalence]
\label{lem:canonical-query-equivalence}
For any supported relation mapped to $X_i \geq x_q \wedge Y_i \leq y_q$, the raw transformed query $(x_q,y_q)$ and the canonical query $(x_q^{+},y_q^{-})$ induce the same valid set.
\end{lemma}

\begin{proof}
By definition of $x_q^{+}$, no transformed $X$-coordinate lies strictly between $x_q$ and $x_q^{+}$.
Thus $X_i\geq x_q$ holds for an object if and only if $X_i\geq x_q^{+}$ holds.
Similarly, no transformed $Y$-coordinate lies strictly between $y_q^{-}$ and $y_q$, so $Y_i\leq y_q$ holds if and only if $Y_i\leq y_q^{-}$ holds.
The raw and canonical transformed queries therefore select the same objects.
\end{proof}

Canonicalization is performed after the semantic transformation, so signed mappings need no special case.
If $Y_i=-t_i$ and $y_q=-t_q$, taking the predecessor of $y_q$ in transformed $Y$ order correctly corresponds to the required boundary in the original $t_i$ order.
Consequently, the number of canonical states is at most $|U_X|\cdot |U_Y|\leq n^2$.
Some states may be empty under a particular workload, but the product grid gives a clean contract for the index.
The next section shows how UDG avoids materializing one proximity graph per canonical state by storing query-state ranges on graph edges.

\section{Unified Dominance Graph}
This section turns the dominance abstraction from Section~\ref{sec:algebraic-mapping-rule} into a graph index.
We fix one relation mapping and construct one UDG instance in the corresponding transformed dominance space. The following construction and search algorithms are relation-independent after this mapping step. Other relations are supported under the same algorithms with a different endpoint transformation.
We first define this query-time contract, then describe UDGConstruction, and finally analyze the structural guarantee.

\subsection{UDG Index and Query}
\label{sec:udg-index-search}

For a mapped query, let $(a,c)=(x_q^{+},y_q^{-})$ be its canonical state, where $a$ is the active $X$ threshold and $c$ is the active $Y$ boundary.
The valid set is
$V(a,c)=\{i \mid X_i \geq a \wedge Y_i \leq c\}$.
UDG uses a single dominance-search abstraction for the mapped relation. After the semantic mapping step, containment, overlap, and signed-coordinate relations differ only in how their endpoints are mapped to $(X,Y)$; the graph-search semantics are identical.

Each node represents a data vector.
For a source node $u$, its adjacency list contains tuples $(l,r,v,b,e)\in G[u]$, where $v$ is a neighboring node.
The tuple is active for $(a,c)$ iff $a\in[l,r]$ and $c\in[b,e]$.
All label values are canonical transformed coordinates from $U_X$ and $U_Y$.
Thus, for containment the labels are interpreted over $(s_i,t_i)$, for overlap over $(t_i,s_i)$, and for signed mappings over values such as $-s_i$ or $-t_i$.
The active subgraph for $(a,c)$ contains exactly the edges whose label rectangle contains that canonical state.

Figure~\ref{fig:udg-state-activation} gives a compact example under containment filtering.
The query interval $[4,13]$ maps to canonical state $(a,c)=(4,13)$, and the left panel shows the corresponding dominance region, where $v_8$, $v_9$, and $v_{10}$ are valid.
The right panel applies the same state to the labeled graph: solid edges are active because their label rectangles contain $(a,c)$, while dashed edges are inactive for this state.
Thus UDG obtains the state-specific subgraph by label tests rather than by materializing a separate graph for each query state.
\begin{figure}[t]
\centering
\subfloat[Dominance state.]{%
    \includegraphics[width=0.46\linewidth]{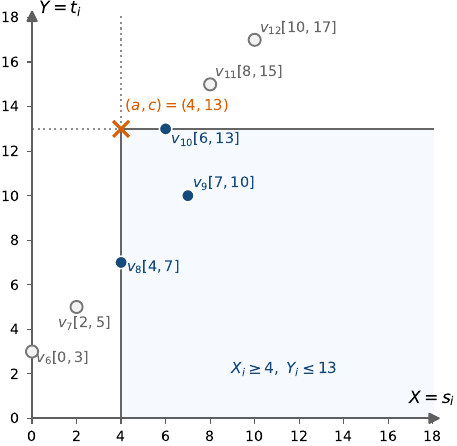}
    \label{fig:udg-dominance-state}
}
\hfill
\subfloat[Edge activation.]{%
    \includegraphics[width=0.50\linewidth]{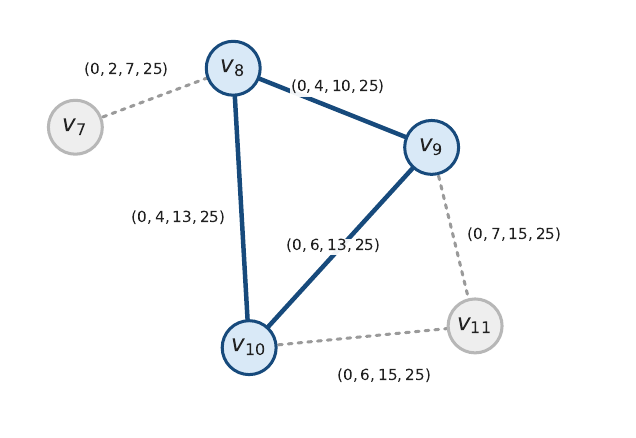}
    \label{fig:udg-edge-activation}
}
\caption{UDG query-state view for containment. Solid graph edges are active at $(a,c)=(4,13)$; dashed edges are inactive.}
\label{fig:udg-state-activation}
\end{figure}

UDG only requires a valid entry point for the canonical state $(a,c)$.
In practice, entry points can be maintained by a compact auxiliary table.
If no valid entry point exists, then $V(a,c)$ is empty.

Algorithm~\ref{alg:udg-search} shows the edge-filtered graph search used by UDG.
For a user query, UDG first maps $\rho$ and $[s_q,t_q]$ to raw transformed coordinates and canonicalizes them to $(a,c)$.
The search then follows the standard graph-based ANNS procedure, except that an outgoing edge is explored only when its label contains $(a,c)$.
The distance function $\delta(\cdot,\cdot)$ is always computed on the original embedding vectors; coordinate transformation affects only filter coordinates and edge activation.

\begin{algorithm}[ht]
\caption{\textsc{UDGSearch}$(G,q,a,c,ep,K)$}
\label{alg:udg-search}
\KwIn{$G$: UDG index; $q$: query vector; $(a,c)$: canonical state; $ep$: valid entry point; $K$: search-pool size}
\KwOut{$\mathit{ann}$: up to $K$ approximate nearest-neighbor candidates of $q$ in state $(a,c)$}

Mark $ep$ as visited\;
$pool \leftarrow \{ep\}$\;
$ann \leftarrow \{ep\}$\;
\While{$pool$ is not empty}{
    $v \leftarrow$ the vector closest to $q$ in $pool$, and pop $pool$\;
    $u \leftarrow$ the vector farthest from $q$ in $ann$\;
    \lIf{$|ann|=K$ and $\delta(v,q)>\delta(u,q)$}{\textbf{break}}
    \ForEach{unvisited $o$ with $(l,r,o,b,e)\in G[v]$}{
        \If{$a \in [l,r]$ and $c \in [b,e]$}{
            Mark $o$ as visited\;
            $u \leftarrow$ the vector farthest from $q$ in $ann$\;
            \If{$|ann|<K$ or $\delta(o,q)<\delta(u,q)$}{
                push $o$ to $pool$ and $ann$\;
                \lIf{$|ann|>K$}{pop the vector farthest from $q$ in $ann$}
            }
        }
    }
}
\Return $ann$\;
\end{algorithm}

The algorithm contains no relation-dependent step after canonicalization.
The label-containment test is the only filter test performed during traversal.
Section~\ref{sec:exact-construction-setting} explains how UDG constructs labels so that every active edge connects valid objects in $V(a,c)$, making this test sufficient for enforcing the interval filter during graph search.

\subsection{UDGConstruction}
\label{sec:exact-construction-setting}

We next describe \textsc{UDGConstruction}, the construction procedure that creates dominance labels.
The construction follows the insertion style of graph-based ANNS indexes: for each new object, it searches the current graph for candidate neighbors, prunes them to obtain diversified neighbors, and stores edges.
The difference is that UDG attaches each emitted edge to a rectangle of canonical query states, so one labeled edge can represent its copies across many state-specific graphs.

Fix a relation mapping and the transformed coordinates from Section~\ref{sec:algebraic-mapping-rule}.
Let $\mathcal{D}^{Y}=(v_1,\ldots,v_n)$ be the objects sorted by increasing transformed $Y$ value with deterministic tie breaking.
We use $U_X[i]$ to denote the $i$-th smallest canonical $X$ threshold in $U_X$.
UDGConstruction scans objects once in this transformed $Y$ order.
For any fixed $Y$ boundary $c$, valid objects form a prefix of this order.
When inserting $v_j$, every object already in the current graph belongs to the previous $Y$ prefix, so the construction only needs to decide for which canonical $X$ thresholds the new proximity edges should be active.

Algorithm~\ref{alg:udg-construction} constructs the compressed graph directly.
For each inserted object $v_j$, the inner loop starts from the current left endpoint $x_L$ and considers only thresholds with $x_L\leq X(v_j)$, because larger thresholds would exclude $v_j$ itself.
At state $(x_L,Y(v_{j-1}))$, UDG searches the current active graph for candidates of $v_j$.
The valid entry point is chosen from the current prefix, so the search starts from an already inserted object and never uses future nodes.
After obtaining $ann$, UDG sets
$x_R=\min(\{X(v_j)\}\cup\{X(u)\mid u\in ann\})$.
For every threshold between $x_L$ and $x_R$, the new object and all returned candidates still satisfy the transformed lower-bound condition.
UDG therefore emits one $X$ label interval $[x_L,x_R]$ for the pruned neighbors selected from $ann$, and then sets the next $x_L$ to the first canonical $X$ value strictly larger than $x_R$.
This leap is the compression step: consecutive canonical states that share the same local construction candidates are represented by one labeled edge interval.

\begin{algorithm}[t]
\caption{\textsc{UDGConstruction$(\mathcal{D}^{Y},U_X,M)$}}
\label{alg:udg-construction}
\KwIn{$\mathcal{D}^{Y}=(v_1,\ldots,v_n)$ sorted by increasing transformed $Y$;
      $U_X$ sorted by increasing transformed $X$; $M$; the max degree}
\KwOut{$G$: UDG index}

$G\leftarrow$ empty labeled graph over $\mathcal{D}$\;
\For{$j\leftarrow 2$ \KwTo $n$}{
    $i\leftarrow 1$\;
    \While{$i\leq |U_X|$}{
        $x_L\leftarrow U_X[i]$\;
        \lIf{$x_L>X(v_j)$}{\textbf{break}}
        $v_{\mathrm{ep}}\leftarrow$ a valid entry point for state $(x_L,Y(v_{j-1}))$\;
        \lIf{$v_{\mathrm{ep}}=\emptyset$}{\textbf{break}}
        $ann\leftarrow \scalebox{0.8}{$\textsc{UDGSearch}(G,v_j,x_L,Y(v_{j-1}),v_{\mathrm{ep}},M)$}$\;
        $x_R\leftarrow \min(\{X(v_j)\}\cup\{X(u)\mid u\in ann\})$\;
        \ForEach{$u\in \textsc{Prune}(v_j,ann,M)$}{
            add $(x_L,x_R,u,Y(v_j),Y(v_n))$ to $G[v_j]$\;
            add $(x_L,x_R,v_j,Y(v_j),Y(v_n))$ to $G[u]$\;
        }
        Set $i$ to the first index with $U_X[i]>x_R$; if none exists, \textbf{break}\;
    }
}
\Return $G$\;
\end{algorithm}

Each emitted tuple has the label rectangle $[x_L,x_R]\times[Y(v_j),Y(v_n)]$ in the transformed canonical space.
The $X$ interval records the thresholds for which the current construction candidates remain valid.
The $Y$ interval records when the insertion of $v_j$ becomes visible to a canonical state.
If a query has canonical boundary $c\geq Y(v_j)$, then $v_j$ is included on the $Y$ side; every linked previous object also satisfies $Y(u)\leq Y(v_j)$ by the construction order.
If $c<Y(v_j)$, the edge is inactive, matching a graph state in which $v_j$ has not yet been inserted.

All orders and labels are defined in transformed coordinates.
Therefore signed mappings require no special handling.
For example, when $Y_i=-t_i$, scanning by increasing $Y_i$ processes larger original end points first, which is exactly the order induced by the normalized condition $Y_i\leq y_q$.
Similarly, the sweep over $U_X$ is always a sweep toward more restrictive lower-bound thresholds, regardless of which original endpoint or sign produces $X_i$.

\subsection{Lossless Compression Analysis}
\label{sec:lossless-compression-analysis}

We now state the structural lossless-emulation guarantee for UDGConstruction.
For each canonical state $(a,c)$, the reference graph is the proximity graph that would be built directly on $V(a,c)$ using the same insertion order and pruning rule.
The guarantee is edge-level: under the construction-time search condition stated in Theorem~\ref{thm:structural-lossless-emulation}, the active UDG subgraph equals this reference graph.
It does not assert exact query-time nearest-neighbor answers.

For a canonical state $(a,c)$, let $G_{\mathrm{UDG}}(a,c)$ be the active subgraph obtained by keeping exactly the edge tuples $(l,r,v,b,e)$ such that $a\in[l,r]$ and $c\in[b,e]$.
Let $G_{\tau}(a,c)$ denote the dedicated insertion-only graph built directly on $V(a,c)=\{i\mid X_i\geq a \wedge Y_i\leq c\}$, using the same $Y$-ordered insertion sequence and the same deterministic \textsc{Prune} rule.

\begin{lemma}[Edge Validity]
\label{lem:edge-validity}
For any edge label emitted by Algorithm~\ref{alg:udg-construction}, if the edge is active for canonical state $(a,c)$, then both endpoints of the edge belong to $V(a,c)$.
\end{lemma}

\begin{proof}
Consider a label emitted while inserting object $v_j$ at threshold $x_L$.
For each pruned neighbor $u$, UDG emits labels with $X$ interval $[x_L,x_R]$ and $Y$ interval $[Y(v_j),Y(v_n)]$, where $x_R=\min(\{X(v_j)\}\cup\{X(w)\mid w\in ann\})$.
If the edge is active for $(a,c)$, then $x_L\leq a\leq x_R$ and $Y(v_j)\leq c$.
Since $x_R\leq X(v_j)$, object $v_j$ satisfies $X(v_j)\geq a$ and $Y(v_j)\leq c$.
Since $u\in ann$, we also have $x_R\leq X(u)$, and hence $X(u)\geq a$.
Moreover, $u$ was inserted before $v_j$, so $Y(u)\leq Y(v_j)$ under the fixed order.
Thus $Y(u)\leq c$.
Both endpoints satisfy the dominance predicate.
The same argument applies to the reverse directed label.
\end{proof}

We analyze our algorithm under the Accurate Search Assumption (ASA), i.e.
during construction, each search returns the exact nearest neighbors for the queried canonical state. We make this assumption because if otherwise, we could not accurately assess the impact of the approximation of the nearest neighbor search to the index structure. The same assumption is used in prior works of range-filtering ANNS~\cite{SeRF-2024,dynamic_serf}, where the lossless-compression is also achieved under ASA assumption.

\begin{theorem}[Structural Lossless Emulation]
\label{thm:structural-lossless-emulation}
Under ASA, for every canonical state $(a,c)\in U_X\times U_Y$,
the active UDG subgraph $G_{\mathrm{UDG}}(a,c)$ is edge-identical to
the dedicated graph $G_{\tau}(a,c)$. In other words, the graph produced by Algorithm~\ref{alg:udg-construction} is a structural lossless compression of the canonical graph family.
\end{theorem}

\begin{proof}
Fix a canonical state $(a,c)$. We prove the claim by induction over the
$Y$-ordered insertion sequence $\tau$. After any processed prefix $P$,
the active UDG subgraph at $(a,c)$ induced by $P$ is identical to the
graph obtained by running the dedicated constructor on
$P\cap V(a,c)$.

The base case is empty. Consider the next object $v$. If $v\notin V(a,c)$,
then the dedicated constructor skips $v$. UDG also creates no active edge
for this state: if $Y(v)>c$, every label emitted for $v$ starts after $c$
on the $Y$ dimension; if $X(v)<a$, every emitted $X$ interval ends no later
than $X(v)$ and therefore cannot contain $a$. Thus the invariant is unchanged.

Now suppose $v\in V(a,c)$. The previous valid candidates for inserting $v$
are exactly the earlier objects $u$ with $X(u)\geq a$; the $Y$ condition is
automatic because all earlier objects have $Y(u)\leq Y(v)\leq c$. By the
induction hypothesis, the active UDG subgraph before inserting $v$ is the
same as the corresponding dedicated graph over these candidates.

Let $[a_0,x_R]$ be the leap interval generated by UDGConstruction that
contains $a$. The construction sets
$x_R=\min(\{X(v)\}\cup\{X(u)\mid u\in ann\})$, so every object in $ann$
remains valid for all thresholds in this interval. Moving the threshold from
$a_0$ to any $a\in[a_0,x_R]$ can only remove objects outside $ann$ and cannot
introduce new candidates. Under ASA, the candidate set used for $a_0$ is
therefore the same candidate set that the dedicated constructor would use
for $a$.

Since \textsc{Prune} is deterministic, both constructions select the same
neighbors of $v$. UDG labels the corresponding directed edges with a rectangle
that contains $(a,c)$, so these edges are active for the fixed state. The leap
intervals are disjoint and cover the canonical thresholds processed for $v$,
so no additional active edge can appear for $(a,c)$.

Thus the induction invariant holds after inserting $v$. After all objects are
processed, $G_{\mathrm{UDG}}(a,c)$ and $G_{\tau}(a,c)$ have the same directed
edge set. Since $(a,c)$ was arbitrary, the theorem follows.
\end{proof}

The theorem gives the precise meaning of structural lossless emulation in UDG: edge labels compress the canonical graph family without changing the active edge set for any canonical state. The proof uses transformed coordinates. Thus containment, overlap, and signed-coordinate mappings follow the same argument after the semantic mapping step.

In the worst case, one inserted object can generate $O(n)$ disjoint $X$-threshold intervals, each emitting at most $O(M)$ directed labels after pruning.
Thus construction has worst-case index size $O(n^2M)$ and index time $O(n^2M^2)$ excluding the online search.

\begin{theorem}
Assuming interval coordinates are independent of vector-space neighbor selection, the expected total number of rounds in Lines~4--14 of Algorithm~\ref{alg:udg-construction} is $O(n\log n)$. Consequently, in the average case, UDG has index size $O(nM\log n)$, makes $O(n\log n)$ calls to the online search in Algorithm~\ref{alg:udg-search}, and takes $O(nM^2\log n)$ indexing time excluding online search, assuming $O(1)$ entry-point lookup. 
\end{theorem}
\begin{proof}
Consider the insertion of $v_j$, and let $z_j$ be the number of rounds of the while loop for this insertion. Let $P_j=\{v_1,\ldots,v_j\}$ be the current prefix in the transformed $Y$ order. Sort $P_j$ by transformed $X$ coordinate and write the order as $w_1,\ldots,w_j$. Let $r_j$ be the rank of $v_j$, i.e., $w_{r_j}=v_j$. The loop only considers thresholds no larger than $X(v_j)$, and therefore only ranks $p\le r_j$. For a rank $p<r_j$, under the independence assumption, the transformed $X$ order is independent of vector-space neighbor selection. Therefore, when there are $t=j-p$ valid previously inserted objects under the current threshold, the boundary object $w_p$ is equally likely to appear at any position in the neighbor ranking and is contained in the candidate set returned by the online search with probability at most $\min\{1,M/t\}$. Let $C_p=\{w_q\mid q\ge p,\ w_q\ne v_j\}$ be the valid previously inserted objects under the threshold starting at rank $p$. Then $|C_p|=j-p$. If the boundary object $w_p$ is not returned by the online search, the minimum $X$ coordinate among $v_j$ and the returned candidates is larger than $X(w_p)$, so the update of the threshold skips rank $p$. Thus rank $p$ can create an additional round only when $w_p$ appears in the returned candidate set. Let $I_p$ be the indicator of this event. The final round, where the threshold reaches $X(v_j)$, contributes at most one additional round. Hence $z_j\le 1+\sum_{p=1}^{r_j-1} I_p$. By the argument above, $\Pr[I_p=1]\le \min\{1,M/(j-p)\}$. Therefore, $\mathbb{E}[z_j]\le 1+\sum_{p=1}^{r_j-1}\min\{1,M/(j-p)\}\le 1+\sum_{t=1}^{j-1}\min\{1,M/t\}\le 1+M+M\sum_{t=M+1}^{j-1}1/t=O(M\log j)$. Since $M$ is a fixed index parameter independent of $n$, $\mathbb{E}[z_j]=O(\log j)$. Summing over all insertions gives $\mathbb{E}[\sum_{j=2}^{n}z_j]=\sum_{j=2}^{n}O(\log j)=O(n\log n)$. Each round calls the online search once, performs one $O(M^2)$ pruning step, and inserts at most $O(M)$ directed edge labels. Therefore, the expected number of online-search calls is $O(n\log n)$, the expected index size is $O(nM\log n)$, and the indexing time excluding online search is $O(nM^2\log n)$. This completes the proof. 
\end{proof}

\section{Practical Unified Dominance Graph}
\subsection{Practical Construction Optimizations}
\label{sec:practical-construction-heuristics}

The exact constructor in Algorithm~\ref{alg:udg-construction} sweeps many canonical \(X\) thresholds for each inserted object \(v\). At each new round, it performs a state-specific UDG search, computes the next sweep boundary from the returned candidate set, and then applies the HNSW-style pruning rule. 
Although the exact constructor provides the structural compression guarantee, running a state-specific search at every canonical threshold is expensive in practice. 
We therefore adopt construction optimizations, following recent RFANNS indexes~\cite{SeRF-2024,dynamic_serf}, to reduce the indexing cost of UDG.  

The first optimization reuses one construction candidate pool. We replace the excessive number of \textsc{UDGSearch} procedure with a single broad search \textsc{UDGSearch($G,v_t, -\inf, \inf, ep, Z)$}, for every new data point $v_t$ inserted. $Z$ is a parameter to find a set $ann$ of $Z$ approximate nearest neighbors of $v_t$. For the canonical X boundary \(x_L\), UDG scans $ann$ in increasing distance to \(v_t\) and keeps only candidates \(u\) satisfying \(X(u) \ge x_L\). The filtered candidates are then passed to \(\mathrm{PRUNE}\).  If no candidate exists in $ann$, we simply terminate the while condition and process the next insertion.

The second optimization increases the leap size along the $X$ dimension.
After pruning the valid candidates at threshold $a$, practical UDG advances according to the pruned neighbor set rather than the full candidate pool.
The conservative policy leaps to the leftmost pruned neighbor.
This reduces the number of construction states while keeping a shared $X$-label interval valid for all emitted neighbors.
The aggressive policy leaps toward the rightmost pruned neighbor, which further reduces construction time and index size but makes the construction more approximate.

UDG further adopts a MaxLeap variant. Let $N$ be the retained neighbor set after pruning. Instead of advancing to the smallest transformed $X$ in
$N$, MaxLeap advances the sweep to
$x_{\mathsf{leap}}=\max_{u\in N}X(u)$. To preserve interval validity, each edge $(v,u)$ is still labeled up to its own valid right boundary, $\min\{X(v), X(u),x_{\mathsf{leap}}\}$.  This larger leap reduces both the number of construction rounds and the number of edges.

These optimizations reduce index time and index size, but bring risk.
The main risk is under-connected active subgraphs in small-valid-set states, where graph search has fewer routing choices.
This motivates the patch edge mechanism below.
\subsection{Validity-Preserving Patch Edges}
\label{sec:patch-edges}

The construction optimizations in Section~\ref{sec:practical-construction-heuristics} reduce index size and construction time by reusing candidate pools and taking larger leap steps.
During this process, we observe a failure mode that is especially important for restrictive filters, i.e., queries that filter out most objects and leave only a small active valid subgraph.
For an inserted object $v$, the sweep is likely to stop before 
$v$ has sufficient neighbors for all valid canonical $X$ thresholds.
We call the remaining consecutive thresholds an uncovered range.
Although the existing edges are still used only under their valid labels, the active graph in this uncovered range may become poorly connected, causing suboptimal search performance under restrictive filters.

UDG addresses this issue with patch edges.
Consider one uncovered range $[a_L,a_R]$ after point $v$ is inserted.
UDG first builds a repair pool from previously inserted objects that are valid at the beginning of the range.
In the transformed dominance space, this condition is simply $X_u\geq a_L$; under containment this follows the start-endpoint order, and under overlap it follows the end-endpoint order.
Thus the repair pool is relation-independent after semantic mapping, while still selecting only candidates that can be active inside the uncovered range.

The repair pool is capped by $M\cdot K_{\mathsf{p}}$, where $K_{\mathsf{p}}$ is a parameter to control the patch candidate budget.
UDG computes the vector distance from $v$ to every candidate in the pool, but distance alone is not sufficient.
A nearest candidate may expire quickly as the canonical $X$ threshold increases, so using only nearest candidates can fail to repair the entire uncovered range.
Therefore, UDG first reserves up to two lifetime anchors, chosen by the largest lifetime rank, ignoring the distance.
These anchors preserve long-lived routing choices across the uncovered range.

The remaining patch slots are filled from the non-anchor candidates in increasing distance to $v$.
UDG applies the same HNSW-style diversity rule used in Lines~4--9 of Algorithm~\ref{alg:prune} to this metric part of the repair set. If there are fewer than $M$ patch neighbors after pruning, UDG backfills with the nearest remaining candidates from the repair pool.

Each selected neighbor $u$ is attached to the uncovered range with right boundary $r=\min\{X_v,X_u,a_R\}$ and a $Y$ label starting from $Y_v$. UDG adds a patch edge $(a_L,r,u,Y_v,Y(v_n))$ to $v$ with a corresponding reverse edge.
The patch edge is active for the canonical query state $(a,c)$, where $\min\{X_v,X_u\}\geq r\geq a$ and $Y_v\leq c$.
Since $u$ is a previously inserted valid candidate, $Y_u\leq Y_v\leq c$.
Both endpoints belong to $V(a,c)$. Edge-filtered search still traverses only objects satisfying the mapped dominance predicate.

Patch edges add routing choices for the uncovered range in the construction.
In the implementation, each object has at most one remaining uncovered range per transformed relation instance, and each range adds at most $M$ patch neighbors.
The additional index time is bounded by the capped repair pool.
Thus patch edges add $O(nM)$ directed edges, which does not change the average bound of UDG and is also dominated by the worst-case construction bound.

\section{Experiment}
This section evaluates UDG on standard ANNS benchmarks and interval-predicate workloads.
\subsection{Experimental Setup}
\label{sec:experimental-setup}

\textbf{Datasets.}
We use public vector datasets from ANN benchmarks~\cite{ann-benchmark}, including SIFT1M~\cite{sift-data}, DEEP1M~\cite{deep1M}, and DBpedia-OpenAI~\cite{dbpedia-entities-openai-1M}.
For these benchmark datasets, we assign synthetic closed intervals over a normalized endpoint domain of size $T$.
For the main synthetic setting, we generate Uniform interval metadata by sampling interval lengths uniformly up to a maximum of $0.01T$ and sampling starts uniformly over the feasible range conditioned on the sampled length; all methods use the same generated interval files, query workloads are selected by exact-count selectivity buckets, and real-world intervals are left uncapped.
We also include two real-world interval workloads: S\&P 500 stock ranges~\cite{hi-png-2025} and an event-time workload extracted from Nasdaq news using an LLM~\cite{ta-rag-2025,dong2024fnspid}.
For these workloads, we keep the real embeddings and data intervals, sample query vectors from the test embeddings, and synthesize query intervals by target selectivity for fair comparison.
Table~\ref{tab:dataset_stats} reports the dataset statistics.
Euclidean distance is used to measure vector similarity.

\begin{table}[t]
\centering
\caption{Dataset statistics.}
\begin{tabular}{@{}lrrr@{}}
\hline
Dataset & Base & Query & Dim. \\
\hline
SIFT1M & 1,000,000 & 10,000 & 128 \\
DEEP1M & 990,000 & 10,000 & 96 \\
DBpedia-OpenAI & 990,000 & 10,000 & 1536 \\
S\&P 500 & 1,445,788 & 14,603 & 384 \\
Nasdaq & 350,466 & 10,000 & 768 \\
\hline
\end{tabular}
\label{tab:dataset_stats}
\end{table}

\textbf{Workloads and query generation.}
We evaluate two primary interval predicates: containment, $s_i\geq s_q \wedge t_i\leq t_q$, and overlap, $t_i\geq s_q \wedge s_i\leq t_q$.
For each query vector, we generate an interval whose valid-set ratio matches a target selectivity $\sigma$.
Controlling selectivity directly is important because the same interval width can produce very different valid-set sizes under different endpoint distributions.
The main selectivity experiment sweeps $\sigma\in\{0.1\%,1\%,5\%,10\%,50\%\}$; later experiments use focused selectivity settings for their specific purpose.

\textbf{Baselines.}
We compare UDG with several hybrid-search ANNS methods.
Hi-PNG~\cite{hi-png-2025} is included for containment queries because it is a containment-oriented interval index.
PostFilter-HNSW first searches a global HNSW index and then removes candidates that do not satisfy the interval predicate~\cite{HNSW-2020}.
PreFilter builds a range tree over interval attributes; at query time, it enumerates the exact valid set and scans the valid vectors.
ACORN is a predicate-agnostic graph index for hybrid search, and we adapt it by treating each interval relation as the predicate during graph traversal~\cite{acorn-2024}.
We use $\gamma=12$ as recommended for filtered search.

\textbf{Metrics.}
We evaluate search quality by Recall@10, $\frac{|R\cap G|}{|G|}$, where $R$ is the retrieved result set and $G$ is the exact top-10 ground truth.
Query efficiency is measured by Queries Per Second (QPS). For each method, we sweep its method-specific query-time parameters and report the resulting Pareto frontier. Points outside the displayed plotting window are omitted for readability. Thus, an absent method in a panel indicates that none of its evaluated configurations yields a comparable operating point within the shown recall--QPS range.

\textbf{Environment.}
We implement the algorithms in C++ compiled with $g{++}$ using the $-\mathrm{Ofast}$ optimization flag and OpenMP parallelism.
All experiments are run on an Ubuntu server with four Intel Xeon Gold 6218 CPUs and 1.5 TB RAM, using 16 threads.
Unless otherwise specified, graph indexes use $M=32$ and construction width efconstruction$=128$, following recent containment-oriented interval ANNS work~\cite{hi-png-2025}.

\subsection{Experiment on Search Performance}
\begin{figure*}[ht]
    \centering
    \includegraphics[width=0.85\linewidth]{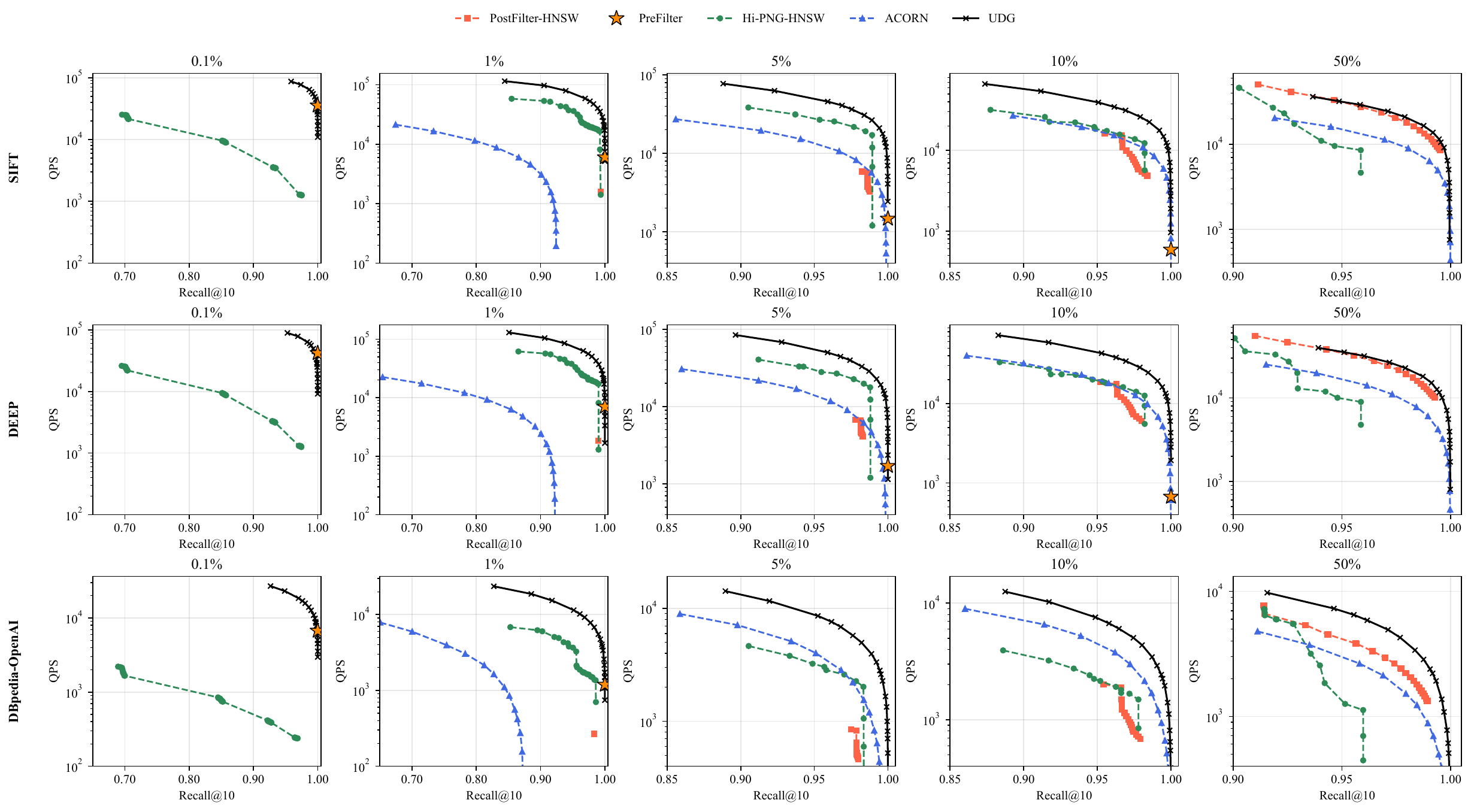}
    \caption{Search performance for containment queries under different selectivities.}
    \label{fig:search-performance-containment}

    \includegraphics[width=0.85\linewidth]{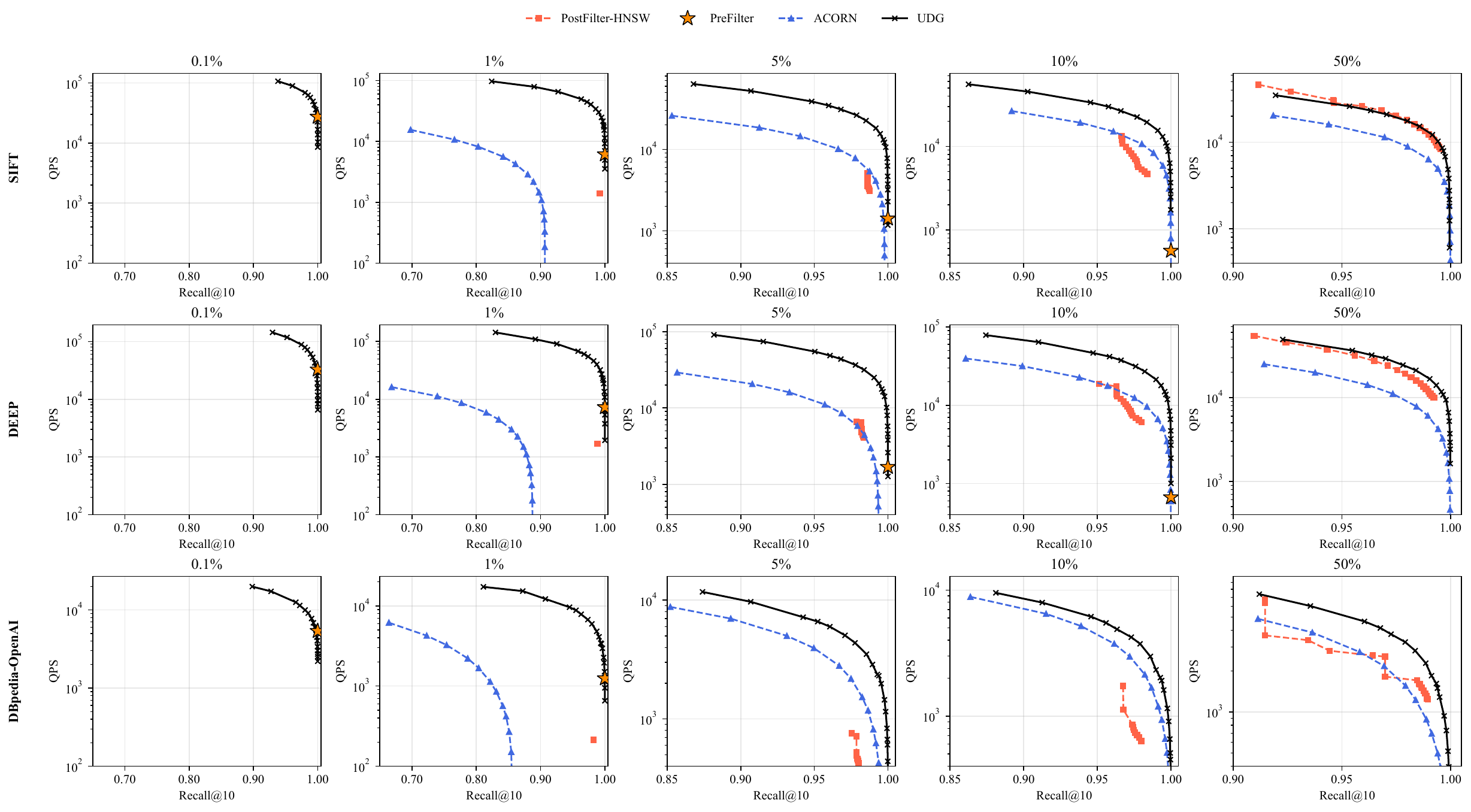}
    \caption{Search performance for overlap queries under different selectivities.}
    \label{fig:search-performance-overlap}
\end{figure*}
\textbf{Main Search Performance}
Figures~\ref{fig:search-performance-containment} and~\ref{fig:search-performance-overlap} report query performance Pareto frontiers for containment and overlap under five target selectivities. Points outside the displayed plotting window are omitted for readability.
Across SIFT1M, DEEP1M, and DBpedia-OpenAI, UDG consistently lies in the high-recall and high-throughput region of the frontier.
In particular, UDG often approaches Recall@10 of 1.0 even at $\sigma=0.1\%$, where the active valid subgraph is extremely small.
This small-$\sigma$ setting is the most sensitive case for UDG because any under-connected valid subgraph immediately hurts recall.

For containment, Figure~\ref{fig:search-performance-containment} shows that UDG usually provides the best or near-best QPS at high recall.
Exact pre-filtering is competitive in the most restrictive cases, as expected, because the valid set is small and can be scanned cheaply.
However, its cost grows with the valid-set size, while UDG keeps a navigable filtered graph and remains competitive across the full selectivity range.
PostFilter-HNSW and ACORN degrade more sharply under restrictive filters because many graph visits or candidate expansions are spent on invalid objects.
Hi-PNG is a strong containment-specific baseline, but UDG achieves comparable or higher throughput at high recall while using a dominance-search operator that also supports other interval predicates.

For overlap, Figure~\ref{fig:search-performance-overlap} uses the same UDG implementation with a different endpoint mapping, $(X_i,Y_i)=(t_i,s_i)$. 
UDG again maintains strong high-recall throughput across selectivities, whereas PostFilter-HNSW and ACORN lose efficiency when the valid set is small.
This result confirms that the gain does not come from a containment-only optimization, but from mapping different two-bound interval predicates into the same dominance-filtered graph search.

% \begin{figure*}[htbp]
\begin{figure*}[ht]
    \centering
    % First Subfigure
    \subfloat[Search performance under real-world interval workloads. \label{fig:sub_realWorld}]{%
        \includegraphics[width=0.9\linewidth]{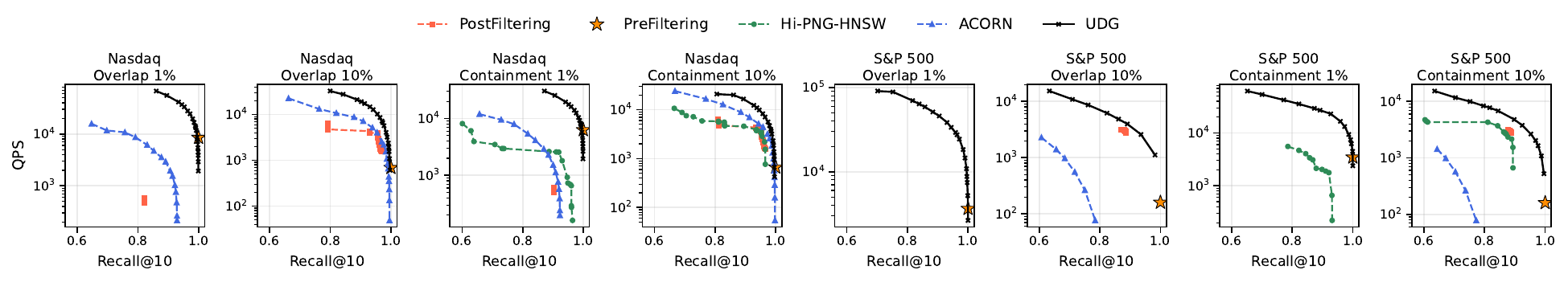}%
    }
    
    % Second Subfigure
    \subfloat[Search performance on additional interval relations. \label{fig:sub_additionInterval}]{%
        \includegraphics[width=0.9\linewidth]{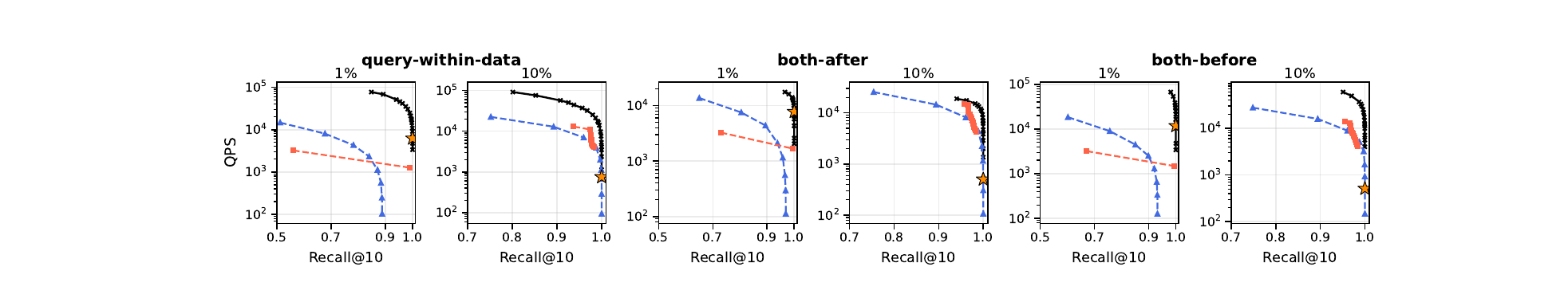}%
    }
    \caption{Search performance on real-world interval workloads and additional interval relations.}
    \label{fig:realWorld_addIntRel}
\end{figure*}
\textbf{Real-world uncapped interval workloads.}
Figure~\ref{fig:sub_realWorld} evaluates UDG on two real-world interval workloads: Nasdaq event-time retrieval and S\&P 500 stock ranges. Unlike the synthetic workloads, these experiments keep the original data intervals without imposing the $0.01T$ length cap, while query intervals are still selected by selectivity buckets for controlled comparison. 
UDG remains effective across the real-world workloads. On both Nasdaq and S\&P 500 datasets, UDG consistently lies on the strongest high-recall frontier for both predicates and selectivities. PreFilter reaches the highest-recall point, but with substantially lower throughput, while UDG still provides a stronger recall--QPS trade-off. These results show that UDG's advantage is not tied to capped synthetic intervals, but extends to real-world interval distributions.

\textbf{Generality beyond the Primary Predicates.}
Figure~\ref{fig:sub_additionInterval} evaluates three additional closed two-bound interval relations on SIFT1M.
This experiment complements the multi-dataset containment and overlap evaluation by checking whether the semantic mapping layer can reuse the same dominance-search operator across different endpoint mappings.
UDG keeps high-recall throughput on these relations, while PostFilter-HNSW and ACORN are less stable when the valid set is small.
The result supports the main abstraction claim that interval semantics can be handled by a shared dominance-search operator after semantic mapping.

\begin{figure}[t]
    \centering
    \includegraphics[width=0.8\linewidth]{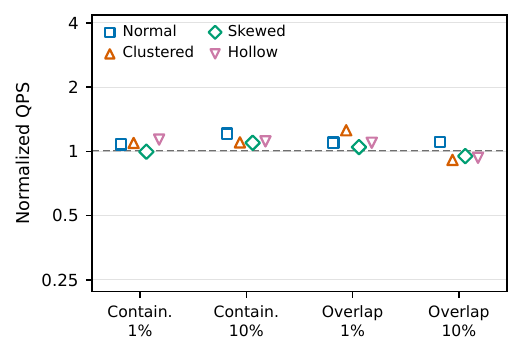}
    \caption{Normalized QPS under interval metadata distributions on SIFT1M (Recall@10 $\geq 0.95$); dashed line: Uniform baseline.}
    \label{fig:interval_distribution_robustness}
\end{figure}
\textbf{Robustness to Interval Metadata Distributions.}
The main experiment uses Uniform interval metadata.
We next vary the interval-attribute distribution on SIFT1M, using Normal, Skewed, Clustered, and Hollow distributions for containment and overlap at 1\% and 10\% selectivity.
Figure~\ref{fig:interval_distribution_robustness} reports QPS normalized by the Uniform workload under the same predicate and selectivity.
UDG remains close to the Uniform baseline across the tested distributions.
Several containment settings are slightly faster than the Uniform workload, while the visible slowdowns are modest and occur mainly for overlap at 10\% selectivity.
Because selectivity is controlled, this comparison isolates the effect of interval-metadata shape.
These results indicate that UDG's efficiency is not an artifact of the Uniform interval distribution.

\subsection{Experiments on Index Construction}
\textbf{Index Time and Space.}
Table~\ref{tab:index_cost} reports construction time (s) and index size (MB) for interval containment on all five datasets; index size excludes raw vector storage. In the table, PostF denotes PostFilter-HNSW and DBP denotes DBpedia-OpenAI.
Containment is used here because Hi-PNG is containment-specific and because different UDG relation mappings may induce slightly different transformed graph structures.
PreFilter is omitted because its range-tree metadata is negligible relative to graph indexes, while its main cost appears during query-time valid-set enumeration.
PostFilter-HNSW has the lowest index cost because it stores a single global graph, but its query performance degrades under restrictive interval filters.
UDG stores dominance labels and patch edges, yet its offline cost remains moderate: it builds faster than ACORN on all datasets and faster than Hi-PNG on the four larger workloads, while its size is comparable to ACORN and below Hi-PNG on the same four larger workloads.
On Nasdaq, Hi-PNG builds faster and uses less memory because this workload has only 350,466 base objects (Table~\ref{tab:dataset_stats}). Hi-PNG recursively partitions the interval space until each leaf contains at most a leaf-size threshold and then builds proximity graphs for tree nodes; the smaller Nasdaq workload therefore reduces the amount of tree-node graph construction.

\begin{table}[t]
\centering
\footnotesize
\caption{Index construction cost}
\setlength{\tabcolsep}{3pt}
\begin{tabular}{@{}llrrrr@{}}
\hline
Dataset & Metric & PostF & ACORN & Hi-PNG & UDG \\
\hline
\multirow{2}{*}{SIFT1M}   & time (s) & 32 & 238 & 337 & 195 \\
                          & size (MB) & 167 & 394 & 1020 & 467 \\
\multirow{2}{*}{DEEP1M}   & time (s) & 30 & 183 & 327 & 177 \\
                          & size (MB) & 175 & 389 & 1090 & 474 \\
\multirow{2}{*}{DBP}      & time (s) & 227 & 1273 & 1272 & 878 \\
                          & size (MB) & 188 & 399 & 1223 & 498 \\
\multirow{2}{*}{S\&P 500} & time (s) & 52 & 493 & 386 & 272 \\
                          & size (MB) & 153 & 569 & 869 & 467 \\
\multirow{2}{*}{Nasdaq}   & time (s) & 28 & 225 & 80 & 152 \\
                          & size (MB) & 51 & 138 & 81 & 151 \\
\hline
\end{tabular}
\label{tab:index_cost}
\end{table}

\begin{figure}[h]
    \centering
    \includegraphics[width=0.9\linewidth]{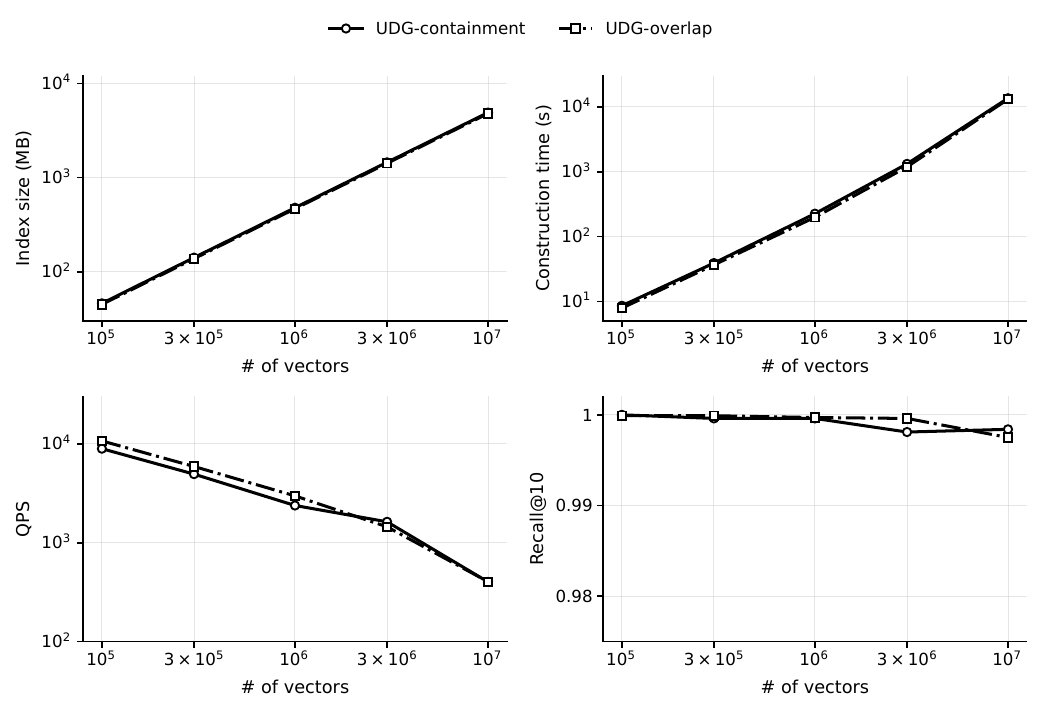}
    \caption{Scalability of UDG on DEEP prefixes with containment and overlap queries.}
\label{fig:scalability}
\end{figure}
\textbf{Scalability.}
Figure~\ref{fig:scalability} evaluates UDG on DEEP prefixes from 100K to 10M base vectors under a fixed-selectivity workload. We fix the parameters $M=32$, $Z=800$, $K_{\mathsf{p}}=8$, and efsearch$=512$, using 16 threads for both index construction and retrieval. 
The index size grows almost linearly with the number of vectors, from about 46MB at 100K vectors to about 4800MB at 10M vectors. Construction time increases substantially with scale, reflecting the offline cost of building the dominance-labeled graph.
At query time, throughput decreases with scale, from roughly 10K QPS to about 400 QPS, but UDG maintains high Recall@10 $> 0.997$ for both containment and overlap.
These results indicate that UDG preserves high search quality up to 10M vectors, with predictable storage growth and the main scalability cost concentrated in offline construction.

\subsection{Impact of Patch Edges}
\label{sec:patch-edge-ablation}
\begin{figure}[h]
    \centering
    \includegraphics[width=0.9\linewidth]{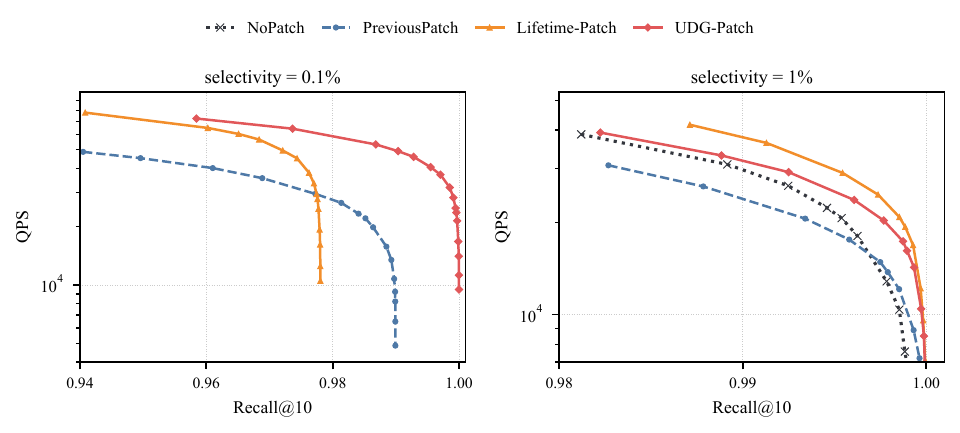}
    \caption{Patch edge ablation with containment queries under different selectivities.}
    \label{fig:patch-edge-search-performance-containment}
    
    \includegraphics[width=0.9\linewidth]{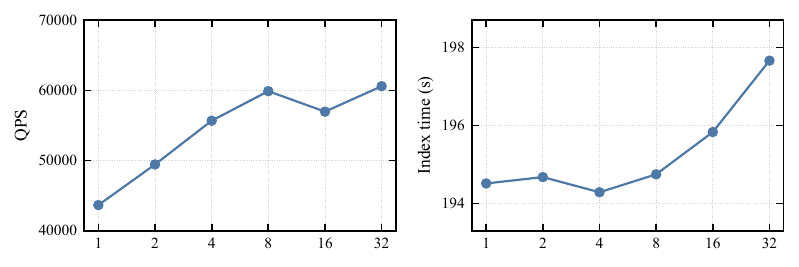}
    \caption{QPS (Recall@10$>0.99$) at 0.1\% selectivity and index time when varying $K_{\mathsf{p}}$.}
    \label{fig:patch_edge_ablation}
\end{figure}
\textbf{Patch Edge ablation.}
Figure~\ref{fig:patch-edge-search-performance-containment} evaluates the effectiveness of patch edges under restrictive interval filters, where only a small fraction of objects remain valid and the active subgraph can easily become disconnected. We compare four variants. NoPatch uses UDG construction without patch edges. PreviousPatch adds repair edges to previous inserted objects, but does not consider neighbor lifetime or distance. Lifetime-Patch selects repair candidates with longer valid lifetimes and applies distance-based diversity pruning, but does not reserve lifetime anchors. UDG-Patch is the full design, which combines lifetime anchors, lifetime-aware candidate selection, and distance-based pruning. For pool-based variants, we use the default pool factor $K_{\mathsf{p}}$=8.

The NoPatch variant reaches only 0.77 recall and is omitted from the 0.1\% selectivity graph. This confirms that patch edges are critical for preserving navigability in small valid-set states. NoPatch degrades because the optimized construction leaves some canonical ranges with too few active routing edges. PreviousPatch improves connectivity, confirming that additional repair edges are useful, but its benefit is limited because many previous neighbors expire quickly as the query state admits fewer valid objects. Lifetime-Patch further improves the query performance frontier by favoring repair edges that remain active over longer ranges, but distance-based pruning can still discard important long-lived anchors, leaving some canonical ranges weakly connected. UDG-Patch achieves the best overall trade-off, showing that lifetime anchors provide stable routing choices and improve search quality in small valid-set states.

We further study the patch pool factor $K_{\mathsf{p}}$ at 0.1\% selectivity. For each $K_{\mathsf{p}}$, we report the best QPS among configurations satisfying Recall@10 $\geq 0.99$, together with the corresponding index time in Figure~\ref{fig:patch_edge_ablation}. Increasing $K_{\mathsf{p}}$ enlarges the repair candidate pool and improves high-recall throughput when the pool is small. The gain, however, saturates for larger pools, while index time increases due to the additional candidate evaluation. We therefore set $K_{\mathsf{p}}=8$ as the default, as it provides high-recall throughput with controlled indexing overhead.

\section{Related Work}
\textbf{Approximate Nearest Neighbor Search.} ANNS is a fundamental primitive for high-dimensional vector retrieval. Due to the curse of dimensionality, exact nearest neighbor search is often too expensive at large scale, so ANNS aims to return high-recall results with much lower query latency~\cite{ANNS-remove-curse-dim-1998}. Existing ANNS indexes can be broadly categorized into hashing-based methods~\cite{gionis1999similaritySearchHash,multi-probe-lsh-2007}, quantization or partition based methods~\cite{PQ-ANNS-2011,ivfpq-multiIndex-2012,optPQ-2014,scalableNNS-cluster-2014,billion-search-gpu-2021}, and graph-based methods~\cite{HNSW-2020,nsg-2019,diskanns-2019,graph-survey-2021}. Among them, graph-based indexes have become widely used because greedy routing on a sparse proximity graph often provides a strong recall--efficiency trade-off. Representative methods include HNSW, which incrementally builds a hierarchical navigable small-world graph, NSG, which refines a candidate graph using a sparsified navigating structure, and DiskANN, which improves scalability for billion-scale vector search~\cite{HNSW-2020,nsg-2019,diskanns-2019}. 

\textbf{Attributed-Filtered ANNS.} Filtered ANNS (or hybrid vector search) retrieves nearest neighbors satisfying structured predicates. 
Existing methods combine vector indexes with pre-filtering, post-filtering, or predicate-aware index~\cite{vbase-2023,filtered-diskann,hqann-2025,acorn-2024,ung-2025}, are flexible for Boolean and categorical metadata but generally fail to guarantee graph navigability within query-dependent valid subsets.
For numeric attributes, range-filtering ANNS exploits one-dimensional scalar ordering to construct range-specific index~\cite{ARKGraph-2023,SeRF-2024,window-filter-2024-RFANNS,iRangeGraph-2024-rfanns,UNIFY-2024,dynamic_serf,DIGRA-2025-rfanns}. While demonstrating the efficacy of filter-aware indexing, these methods are structurally bound to scalar attribute. Consequently, they are not applicable for interval predicates like containment and overlap, which couple two endpoints and induce two-dimensional dominance regions.

Temporal and interval-aware ANNS has also received recent attention. TANNS optimizes point-in-time retrieval using validity intervals and query timestamps~\cite{TANNS-2025}, whereas Hi-PNG studies the containment instance where both data and queries are intervals~\cite{hi-png-2025}. Concurrent range-range ANNS work proposes a labeled multi-segment tree graph~\cite{mstg-2026} for arbitrary range-range predicates. Separately, classical interval indexes~\cite{relational_interval_tree} excel at exact low-dimensional enumeration, but they inherently lack the vector-space proximity structures strictly required for high-dimensional ANNS. In contrast, UDG targets closed two-bound conjunctive interval predicates and unifies them under a single edge-activation rule within a normalized dominance space.

\section{Conclusion}
This paper presented the Unified Dominance Graph (UDG), a cohesive graph indexing framework for Interval-Predicate ANNS. By mapping supported interval predicates into a normalized two-dimensional dominance space, UDG unifies containment, overlap, and related endpoint-bound relations under a single graph-search abstraction. 
UDG achieves structural lossless compression of query-specific proximity graphs, while our proposed patch edge mechanism adds routing choices in restrictive filter states, where the active valid subgraph can become poorly connected.  
Extensive experiments show that UDG achieves superior recall throughput trade-offs across interval relations, while requiring substantially lower indexing cost than specialized interval-filtering indexes and remaining competitive with general hybrid-search baselines.
UDG establishes a versatile, effective foundation for database systems to jointly process ANNS and interval attributed metadata.

% \section*{Acknowledgment}
\section{AI-GENERATED CONTENT ACKNOWLEDGEMENT}
The authors acknowledge the use of ChatGPT to assist with code implementation and language polishing. These tools were used for implementation support and editorial refinement, not for generating scientific claims, experimental results, or analysis. All scientific content, code, experiments, and reported results were reviewed and verified by the authors.

% TODO: use bibitem at final version
% \bibitem{zipfs_uni_power_low}Ectors, W., Kochan, B., Janssens, D., Bellemans, T. \& Wets, G. Exploratory analysis of Zipf’s universal power law in Activity Schedules. {\em Transportation}. \textbf{46}, 1689-1712 (2018,2)

% \end{thebibliography}

\bibliographystyle{IEEEtran}
\bibliography{custom}
\end{document}